\documentclass[a4paper,USenglish,cleveref, autoref, thm-restate]{lipics-v2021}
\usepackage{amssymb,amsmath}
\usepackage{cite}

\title{Beyond Outerplanarity}

\author{Steven~Chaplick}{Maastricht University, the Netherlands and Universit\"at W\"urzburg, Germany}{s.chaplick@maastrichtuniversity.nl}{https://orcid.org/0000-0003-3501-4608}{}

\author{Myroslav~Kryven}{Amherst College, USA and Universit\"at W\"urzburg, Germany}{mkryven@amherst.edu}{https://orcid.org/0000-0003-4778-3703}{}

\author{Giuseppe~Liotta}{Department of Engineering, Universit\`a degli
  Studi di Perugia, Italy}{giuseppe.liotta@unipg.it}{https://orcid.org/0000-0002-2886-9694}{}

\author{Andre~L\"offler}{Universit\"at W\"urzburg, Germany}{info@andre-loeffler.net}{}{}

\author{Alexander~Wolff}{Universit\"at W\"urzburg, Germany \and \url{https://www.informatik.uni-wuerzburg.de/en/algo/team/wolff-alexander}}{}{http://orcid.org/0000-0001-5872-718X}{}

\authorrunning{S.~Chaplick, M.~Kryven, G.~Liotta, A.~L\"offler, and
  A.~Wolff}

\Copyright{Steven~Chaplick, Myroslav~Kryven, Giuseppe~Liotta,
  Andre~L\"offler, Alexander~Wolff}


\relatedversion{A preliminary version of this work has appeared in
  Proc.\ 25th International Symposium on Graph Drawing and Network
  Visualization~\cite{ckllw-bo-GD17}.}

\acknowledgements{We thank Alexander Ravsky, Thomas van Dijk, Fabian
  Lipp, and Johannes Blum for their comments and preliminary
  discussion.  We also thank David Wood for pointing us to several
  important references~\cite{DressKM02,Nakamigawa00,Wood2007}.  We
  are grateful for the very helpful comments of our reviewers.}

\nolinenumbers %uncomment to disable line numbering

\newtheorem{question}[theorem]{Question}

\graphicspath{{figures/}}

\definecolor{defblue}{rgb}{0.121,0.47,0.705}
\let\emph\relax
\DeclareTextFontCommand{\emph}{\color{defblue}\em}


\hideLIPIcs

\begin{document}

\maketitle

\begin{abstract}
We study straight-line drawings of graphs where the vertices are
placed in convex
position in the plane, i.e., \emph{convex drawings}. We consider two
families of graph classes with convex drawings:
\emph{outer $k$-planar} graphs, where each edge is crossed by at most $k$ other
edges; and \emph{outer $k$-quasi-planar} graphs, where no $k$ edges can
mutually cross.

We show that the outer $k$-planar graphs are
$\lfloor3.5\sqrt{k}\rfloor$-degenerate, and consequently that every
outer $k$-planar graph 
can be colored with $\lfloor3.5\sqrt{k}\rfloor + 1$ colors.
We further show that every outer $k$-planar graph has a balanced
vertex separator of size at most $2k+3$.
For each fixed $k$, these small balanced separators allow us to test
outer $k$-planarity in quasi-polynomial time,
e.g., this implies that none of these recognition problems is NP-hard
unless the Exponential Time Hypothesis fails.
We also show that the class of outer 3-quasi-planar graphs and the
class of planar graphs are incomparable.

Finally, we restrict outer $k$-planar and outer $k$-quasi-planar drawings
to \emph{full} drawings (where no crossing appears on the boundary of
the outer face) and to \emph{closed} drawings (where the vertex
sequence on the boundary of the outer face is a Hamiltonian cycle in
the graph).  For each $k$, we express \emph{closed outer $k$-planarity}
and \emph{closed outer $k$-quasi-planarity} in extended monadic
second-order logic.
Since every outer $k$-planar graph has treewidth $O(k)$, Courcelle's
theorem implies that closed outer $k$-planarity is linear-time
testable.  We leverage this result to further show that full outer
$k$-planarity can also be tested in linear time. 
\end{abstract}

\section{Introduction}

A \emph{drawing} of a graph maps each vertex to a distinct point in the plane, each edge to a Jordan curve connecting the points of its incident vertices but not containing the point of any other vertex, and two such Jordan curves have at most one common point.
In the last few years, the focus in graph drawing has shifted from
exploiting structural properties of planar graphs to addressing the
question of how to produce well-structured (understandable) drawings
in the presence of edge crossings, i.e., to the topic of
\emph{beyond-planar} graph classes.
The primary approach here has been to define and study graph classes
which allow some edge crossings, but restrict the crossings in various
ways, such as, for example, restricting the total number of
crossings~\cite{Grohe-JCSS04,kr-ccnilt-STOC07} or even bundling the edges and
restricting the number of crossings of bundles rather than of single
edges~\cite{h-hebva-TVCG06,cdkprw-bcr-GD19,cdkprw-bcr-JGAA2020}.
Two of the most commonly studied such graph classes are:
\begin{enumerate}
\item \emph{$k$-planar graphs}, that is, the graphs that can be drawn
  so that each edge is crossed by at most $k$ other edges.
\item \emph{$k$-quasi-planar graphs}, that is, the graphs that can be drawn so
  that no $k$ pairwise non-incident edges mutually cross.
\end{enumerate}
Note that the $0$-planar graphs and $2$-quasi-planar graphs are
precisely the planar graphs.  The $3$-quasi-planar
graphs are simply called \emph{quasi-planar}.
For example, $K_6$ is 1-planar and, hence, quasi-planar; see Fig.~\ref{fig:intro-non-convex}.
While all $k$-planar graphs are clearly $(k+2)$-quasi-planar, they are
actually even $(k+1)$-quasi-planar~\cite{kp-vs-kqp-2020}.

\begin{figure}
  \begin{subfigure}[b]{0.44\textwidth}
    \centering
    \includegraphics[page=1]{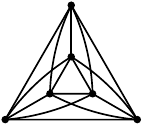}
    \caption{a 1-planar and, hence, quasi-planar drawing}
    \label{fig:intro-non-convex}
  \end{subfigure}
  \hfill
  \begin{subfigure}[b]{0.53\textwidth}
    \centering
    \includegraphics[page=2]{intro-K6}
    \caption{an outer 4-planar and outer 4-quasi-planar drawing}
    \label{fig:intro-convex}
  \end{subfigure}
  \caption{Two drawings of~$K_6$.  Note that $K_6$ is neither outer
    3-planar nor outer quasi-planar.}
  \label{fig:intro}
\end{figure}

% We can similarly generalize...

In this paper we restrict the two above families of graph classes by
insisting on drawings where all vertices lie 
on the outer face and the edges are routed in the complement of the outer face. Note that this is equivalent to placing the vertices
in convex position and drawing the edges as straight-line segments.
The graphs that admit such drawings without crossings are the
\emph{outerplanar} graphs.  In this paper, we apply the above two
generalizations of planar graphs to outerplanar graphs, which yields
the classes of \emph{outer $k$-planar} graphs and \emph{outer
  $k$-quasi-planar} graphs.  For example, $K_6$ is neither outer
3-planar nor outer quasi-planar, but outer 4-planar and,
trivially, outer 4-quasi-planar (because it does not have four
independent edges); see Fig.~\ref{fig:intro-convex}.
% (For an example, see Fig.~\ref{fig:3tree},
% which shows a planar graph that is not outer quasi-planar,
% but removing any of its vertices makes it outer quasi-planar.)
For these classes, we consider balanced separators, treewidth,
degeneracy (see Section~\ref{sub:concepts} below), coloring, edge
density, and recognition.

\subsection{Related Work}

Ringel~\cite{Ringel1965} was the first to consider
$k$-planar graphs; he showed that $1$-planar graphs are
7-colorable.  Twenty years later, Ringel's result was improved by
Borodin~\cite{borodin1984-6col-1p}, who showed that 1-planar graphs
are in fact 6-colorable.  This is tight since $K_6$ is 1-planar.
Many additional results on 1-planarity can be found in a recent survey paper~\cite{1p-biblio-2017}.
Generally, every $n$-vertex $k$-planar graph has treewidth
$O(\sqrt{kn})$~\cite{dujmovic2017structure} and at most
$3.81n\sqrt{k}$ edges~\cite{a-tgfce-CG19}, and hence degeneracy at most
$\lfloor 2 \cdot 3.81\sqrt{k} \rfloor$ and chromatic
number at most $\lfloor 2 \cdot 3.81\sqrt{k} \rfloor + 1$.

Outer $k$-planar graphs have been considered mostly for $k \in \{0,1,2\}$.
Of course, the outer 0-planar graphs are the classic outerplanar
graphs which are well-known to be 2-degenerate and to have treewidth at
most~2. It was shown that essentially every graph property can be
tested efficiently on outerplanar graphs~\cite{babu_et_al:LIPIcs:2016:6644}.
Outer 1-planar graphs are a simple subclass of planar graphs and can
be recognized in linear time~\cite{Auer2016,Hong2015}.
\emph{Full outer $2$-planar graphs}, that is, outer 2-planar graphs
that admit a drawing where no crossing appears on the boundary of the
outer face, can be recognized in linear time~\cite{hn-ltatfo2p-DAM19}.
General outer $k$-planar graphs were considered by Binucci et
al.~\cite{Binucci2016}, who showed (among other results) that, for every $k$,
there is a 2-tree that is not outer $k$-planar.
Wood and Telle~\cite{Wood2007} considered a slight generalization of outer $k$-planar
graphs in their work and showed that these graphs have treewidth~$O(k)$. 
Outer 1-planar graphs have been compared~\cite{b-fcfcfg-IPL18} with \emph{fan-crossing} graphs 
(that is, graphs that have a drawing where each edge can only cross edges with a common endpoint) and
\emph{fan-crossing-free} graphs (that is, graphs that have a drawing
where no edge is crossed by two or more edges with a common endpoint).
Fan-crossing and fan-crossing-free are complementary properties,
in the sense that a drawing is 1-planar if and only if it is
fan-crossing and fan-crossing-free.  Brandenburg~\cite{b-fcfcfg-IPL18}
showed that there are graphs that are simultaneously (outer-)
fan-crossing and (outer-) fan-crossing-free but not (outer-) 1-planar.
Angelini, Da Lozzo, F\"orster, and
Schneck~\cite{alfs-2lkpg-GD20,alfs-2lkpg-CJ23} studied the
edge density of $k$-planar \emph{2-layer} layouts (where the vertices
lie on two horizontal lines and every edge is a y-monotone curve). 
In order to admit such a layout, a graph must be
bipartite, and the constraints on the placement of the vertices
emphasize its bipartite structure.
$k$-Planar layouts have also been studied in the context of
\emph{book embeddings}~\cite{lch-pafob-COCOA20,ackssuw-ecog-SWAT24},
where the vertices lie (possibly in a given fixed order) on a straight
line, the \emph{spine}, and the edges are placed in a given number of
halfplanes (called \emph{pages}) whose pairwise intersection is the spine.

The $k$-quasi-planar graphs have been studied extensively from the
perspective of edge density.
Pach, Shahrokhi, and Szegedy~\cite{Pach1996} conjectured that
every $n$-vertex $k$-quasi-planar graph has at most $c_kn$ edges,
where $c_k$ is a constant depending only on $k$.  Their conjecture has
been proven to hold for $k=3$~\cite{AT07} and $k=4$~\cite{Ackerman2009}. 
The best known general upper bound is
$n(\frac{c\log{n}}{\log{k}})^{2\log{k}-4}$~\cite{Fox_2022,FoxPS-EJC24},
where $c$ is a positive constant. 
Edge density was also considered in the ``outer'' setting:
Capoyleas and Pach~\cite{CP92} showed that every outer $k$-quasi-planar
graph with $n$ vertices has at most $2(k-1)n - \binom{2k-1}{2}$ edges.
Nakamigawa~\cite{Nakamigawa00} and
Dress, Koolen, and Moulton~\cite{DressKM02} showed that there
are outer $k$-quasi-planar graphs meeting this bound (if $n \geq 2k-1$).
Moreover, the outer $k$-quasi-planar graphs that meet this bound are
exactly the maximal outer $k$-quasi-planar graphs. (Recall that a
graph is \emph{maximal} with respect to a given graph property if
adding any edge to the graph destroys the property.)
Both teams of authors actually
showed that, given any two maximal outer $k$-quasi-planar graphs~$G$
and~$G'$ with the same vertex set but different edge sets, for every
two outer $k$-quasi-planar drawings~$\Gamma$ and~$\Gamma'$ of~$G$
and~$G'$, respectively, whose corresponding vertices are in the same
positions, there exists a sequence of local edge exchange operations
(called \emph{flips}) producing drawings
$\Gamma=\Gamma_1, \Gamma_2, \dots, \Gamma_t=\Gamma'$ such that every
intermediate drawing is an outer $k$-quasi-planar drawing.
More recently, it was
shown that the \emph{semi-bar $k$-visibility graphs} are outer
$(k+2)$-quasi-planar~\cite{Geneson_2014}.
Apart from these results, the outer $k$-quasi-planar graphs do not seem to have
received much attention.

The \emph{convex} (or \emph{1-page book}) \emph{crossing number} of a
graph~\cite{schaefer2013graph} is the minimum number of crossings,
taken over all convex drawing.  This concept has been introduced
several times (see~\cite{schaefer2013graph} for more details) and is
closely related to \emph{$k$-quasi-planar geometric graphs}, which are
graphs with straight-line $k$-quasi-planar
drawings~\cite{valtr-pavel}.

The convex crossing number is NP-complete to
compute~\cite{masuda1987np}.  However, Bannister and
Eppstein~\cite{Bannister_2014,be-cm1p2pdgbtw-JGAA18} used
treewidth-based techniques (via
extended monadic second order logic and Courcelle's theorem; see
Theorem~\ref{thm:courcelle}) to show that the convex crossing number
can be computed in linear time.
More precisely, they show that, given a graph $G$ with $n$ vertices
and $m$ edges, the convex crossing number of~$G$,
$\operatorname{cr}^\circ(G)$, can be computed in
$O(f(\operatorname{cr}^\circ(G))\cdot (n+m))$ time, where $f$ is a
computable function.  Thus, recognizing \emph{outer $k$-crossing
  graphs} is \emph{fixed-parameter tractable}.

There is a natural connection between outer $k$-planar graphs and
$(k/2)$-outerplanar graphs.
Consider an outer $k$-planar graph $G$ with some outer $k$-planar
drawing and place a dummy vertex at each crossing to get a planarized
graph~$G'$.  The graph~$G'$ is $(k/2)$-outerplanar.  The treewidth
of~$G'$ is known to be at most $3/2k$~\cite{bodlaender1998}.  We can
construct a tree decomposition of~$G$ from the tree decomposition
of~$G'$ by replacing every dummy vertex~$v$ of~$G'$ in each bag with
the four endpoints of $G$ corresponding to~$v$.  This increases the
treewidth of the decomposition of~$G'$ by a factor of~4.  Therefore,
the treewidth of~$G$ is at most~$6k$.

With more involved arguments, Wood and Telle~\cite{Wood2007} showed
that every outer $k$-planar graph (actually every graph from a
somewhat larger class of graphs) has treewidth at most $3k+11$.
Recently, this upper bound has been improved by Firman, Gutowski,
Kryven, Okada, and Wolff~\cite{fgkow-btokp-GD24} to
$\lfloor 1.5k \rfloor + 2$.
Strengthening the connection between outer $k$-planar graphs and
$O(k)$-outerplanar graphs,
Pyzik~\cite{pyzik2025treewidthouterkplanargraphs} constructed, for
every $k$, a $k$-outerplanar graph of treewidth $3k$ that is also
outer $2k$-planar.

\subsection{Preliminaries}
\label{sub:concepts}

We briefly define the key graph theoretic concepts that we will study.
Given a graph~$G$, let~$V(G)$ denote its vertex set and let~$E(G)$
denote its edge set.

A graph is \emph{$d$-degenerate}~\cite{lick1970k} if every subgraph of it has a vertex of degree at most $d$.
This concept is used in greedy algorithms for coloring. Namely, a $d$-degenerate graph can be inductively $(d+1)$-colored by simply removing a vertex of degree at most $d$.
A graph class is $d$-degenerate if every graph in the class is $d$-degenerate. Furthermore, a graph class which is \emph{hereditary} (i.e., closed under taking subgraphs) is $d$-degenerate when every graph in that class has a vertex of degree at most $d$.
Note that outerplanar graphs are 2-degenerate, and planar graphs are 5-degenerate.

A \emph{separation} of a graph $G$ is a pair $(A,B)$ of subsets of $V(G)$ such that
$A \cup B = V(G)$, and no edge of $G$ has one end in $A\setminus B$ and the other in
$B \setminus A$. The set $A \cap B$ is called \emph{separator}, and
the \emph{size} of the separation $(A, B)$ is $|A \cap B|$.
A separation $(A, B)$ of a graph $G$ on $n$ vertices is \emph{balanced} if
$|A \setminus B| \leq {2n}/{3}$
and $|B \setminus A| \leq {2n}/{3}$.
The \emph{separation number} of a graph $G$ is the smallest
number $s$ such that
every subgraph of $G$ has a balanced separation of size at most~$s$.
The \emph{treewidth} of a graph was introduced by Robertson and
Seymour~\cite{robertson1984graph}.
It is a fundamental graph parameter that measures how close a graph is
to being a tree.  Treewidth plays a crucial role in algorithm design,
as many otherwise intractable problems become efficiently solvable on
graphs of bounded treewidth.  Treewidth is closely related to
separation number.
Namely, any graph with treewidth~$t$ has separation number at most $t+1$
and, as Dvo\v{r}\'ak and Norin~\cite{dn-tgbs-JCTS19} showed, any graph
with separation number~$s$ has treewidth at most $15s$.
% Graphs with bounded treewidth are well-known due to
Furthermore, due to Courcelle's theorem~\cite{courcelle1990}
(which we cite in Theorem~\ref{thm:courcelle}), if a graph class has
bounded treewidth, it means that many problems can be solved
efficiently for graphs in that class.

A \emph{quasi-polynomial time} algorithm is one with a running time
  of the form $2^{\mathrm{polylog}(n)}$, where $n$ is the size of the input.
The \emph{Exponential Time Hypothesis (ETH)}~\cite{IMPAGLIAZZO2001367}
is a complexity theoretic assumption defined as follows. For $k\geq3$,
let
\[s_k=\inf\{\lambda \colon \text{there is an } O(2^{\lambda n})
\text{-time algorithm to solve $k$-SAT}\}.\]
The ETH states that for $k \geq 3$, $s_k >0$.  Hence, for example,
there is no quasi-polynomial time algorithm that solves 3-SAT.
So, finding a problem that can be solved in quasi-polynomial time and 
is also NP-hard, would contradict the ETH.
In recent years, the ETH has become a standard assumption from which many conditional lower bounds have been proven~\cite{Cygan2015}.
Note that, in addition to violating the ETH, the existence of an
NP-hard problem which can be solved in quasi-polynomial time would
also directly imply that \emph{nondeterministic exponential time
  (NEXP)} coincides with \emph{deterministic exponential time (EXP)}
(which can be proven by a padding argument similar
to~\cite[Proposition~2]{BuhrmanH92}).  Thus, having such an algorithm
for a problem implies that it is extremely unlikely for that problem
to be NP-hard.

\subsection{Contribution}

We first consider outer $k$-planar graphs; see Section~\ref{sec:o-kp}.
We show that the largest outer $k$-planar {\em complete} graph has  
$(\lfloor\sqrt{4k+1}\rfloor+2)$ vertices.
Further, we show that every outer $k$-planar graph is
$\lfloor3.5\sqrt{k}\rfloor$-degenerate and hence has chromatic number
at most $\lfloor3.5\sqrt{k}\rfloor + 1$.
Next we show that every outer $k$-planar graph has 
separation number at most $2k+3$.  While this does not improve the
current best bound~\cite{fgkow-btokp-GD24} of $k+2$ on the separator
of outer $k$-planar graphs, we use our separator construction
to obtain a quasi-polynomial time algorithm to
test outer $k$-planarity (for fixed~$k$), which means that
the recognition problem is not NP-hard unless ETH fails.
Very recently, this result has been improved by Kobayashi, Okada, and
Wolff~\cite{kow-r2lok-SoCG25}, who showed that the recognition problem
can be solved in $2^{O(k \log k)}n^{3k+O(1)}$ time, that is, it lies
in the class XP (slicewise polynomial).  They also showed that the
problem is XNLP-hard.  This implies that it is W[$t$]-hard for
every~$t$ and that it is unlikely that it admits a fixed-parameter
algorithm, that is, an algorithm running in $f(k) \cdot n^{O(1)}$ time
for some computable function~$f$.

Then we show that the class of outer quasi-planar graphs and the
class of planar graphs are incomparable; see Section~\ref{sec:o-kqp}.

Finally, we restrict outer $k$-planar and outer $k$-quasi-planar drawings
to \emph{full} drawings, i.e., drawings where no crossing appears on the boundary of
the outer face, and to \emph{closed} drawings, i.e., drawings where the vertex
sequence on the boundary of the outer face is a
cycle in the graph; see Section~\ref{sec:full}.
(Note that every closed drawing is a full drawing.)
The case of full outer 2-planar graphs has been considered by Hong and Nagamochi~\cite{hn-ltatfo2p-DAM19}
who showed that full outer 2-planarity testing can be performed in linear time.
They observed that a graph is full outer $2$-planar if and only if its
maximal biconnected components are closed outer $2$-planar. 
We generalize this observation to $k$-planar and $k$-quasi-planar
graphs, that is, a graph is full outer $k$-planar (full
$k$-quasi-planar) if and only if its maximal biconnected components
are closed outer $k$-planar (full $k$-quasi-planar).
Then, for each $k$, we express \emph{closed outer $k$-planarity}
(and \emph{closed outer $k$-quasi-planarity}) in
\emph{extended monadic second-order logic}.
Thus, since outer $k$-planar graphs have bounded treewidth,
full outer $k$-planar graphs can be recognized in
$O(f(k)\cdot n)$ time, for a computable function~$f$.  In other words,
this recognition problem {\em is} fixed-parameter tractable.
We note that this result greatly generalizes the work of Hong and
Nagamochi~\cite{hn-ltatfo2p-DAM19}.  Our general approach via
Courcelle's theorem is similar to that of Bannister and
Eppstein~\cite{be-cm1p2pdgbtw-JGAA18} for computing the convex
crossing number of a graph.

\section{Outer {\em k}-Planar Graphs}
\label{sec:o-kp}

In this section, we study the structural properties of outer
$k$-planar graphs such as degeneracy and separation number.
Based on these structural properties, we obtain bounds on the colorability  
of outer $k$-planar graphs as well as a quasi-polynomial-time
recognition algorithm.
% We show that every outer $k$-planar graph is
% $\lfloor 3.5\sqrt{k} \rfloor + 1$-degenerate and has separation number $O(k)$. This
% provides tight bounds on the chromatic number, and
% allows for testing outer $k$-planarity in quasi-polynomial~time.

\subsection{Degeneracy}

First, we focus on the degeneracy of outer $k$-planar {\em complete} graphs.
We show that the largest outer $k$-planar complete graph has 
${\lfloor\sqrt{4k+1}\rfloor + 2}$ vertices;
see Lemma~\ref{obs:largest_clique}. This implies that there are
outer $k$-planar graphs whose minimum degree is $\lfloor\sqrt{4k+1}\rfloor + 1$.
We then bound the degeneracy of general outer $k$-planar graphs
by $\lfloor 3.5\sqrt{k} \rfloor$ from above; see Theorem~\ref{thm:maxmindeg}. 

\begin{lemma}
  \label{obs:largest_clique}
  For every $k \ge 0$, the largest outer $k$-planar complete graph has
  $\lfloor\sqrt{4k+1}\rfloor + 2$ vertices. 
\end{lemma}

\begin{proof}
  The largest outerplanar complete graph is~$K_3$, so the statement is
  correct for $k=0$.

  Now let $k \ge 1$ and
  consider an outer $k$-planar drawing of~$K_n$ for some $n \ge 1$.
  Let $e$ be an edge that splits the complete graph
  so that there are $(n-2)/2$ many vertices on both sides if $n$ is
  even, or $(n-1)/2$ on one side and $(n-3)/2$ on the other if $n$ is
  odd.  Then the edge $e$ has the largest number of crossings among
  all the edges in the drawing, namely $(n-2)^2/4=(n^2-4n+4)/4$ if $n$
  is even and $(n-3)(n-1)/4=(n^2-4n+3)/4$ if $n$ is odd.
  Taking into account the fact that no
  edge is crossed more than $k$ times, we obtain that
  $n \le \lfloor 2\sqrt{k} \rfloor +2$ if $n$ is even and
  $n \le \lfloor \sqrt{4k+1}\rfloor + 2$ if $n$ is odd.

  On the other hand, let $n = \lfloor 2\sqrt{k} \rfloor +2$.
  % $n = \lfloor \sqrt{4k + 1}\rfloor + 2$.
  Place the vertices of~$K_n$ on the corners of a convex $n$-gon.
  Since $n \le \lfloor \sqrt{4k+1}\rfloor + 2$, it is clear that the
  resulting convex drawing is outer $k$-planar.  We claim that this is
  even true for $n = \lfloor \sqrt{4k+1}\rfloor + 2$.

  Note that the two bounds deviate only if $4k+1$ is a square number,
  say $\ell^2$.  Then $\lfloor \sqrt{4k+1}\rfloor = \ell$, whereas
  $\lfloor \sqrt{4k} \rfloor = \lfloor \sqrt{\ell^2-1} \rfloor = \ell -1$.
  In this case, however, $n=\lfloor \sqrt{4k+1}\rfloor + 2 = \ell+2$
  is odd, because $4k+1$ is odd.  Hence, the slightly larger bound
  applies, and $K_n$ is outer $k$-planar.
  % If $n$ is odd, then the maximum number of crossings per edge is
  % $(n-3)(n-1)/4=k$, and hence, the resulting drawing is outer
  % $k$-planar.  If $n$ is even, then
  % $\lfloor \sqrt{4k + 1}\rfloor = \lfloor 2\sqrt{k}\rfloor$, and
  % thus, the maximum number of crossings per edge is $(n-2)^2/4 = k$,
  % and hence, the resulting drawing is outer $k$-planar.
\end{proof}

Let $G$ be an outer $k$-planar graph. Consider some outer $k$-planar
drawing of~$G$.  Recall that vertices lie on a
circle and the edges are straight-line segments.
We say that an edge $ab$ \emph{splits off}~$l\in \mathbb{N}$ vertices of $G$ \emph{to one side} if
one of the open half-planes defined by the edge $ab$ contains exactly $l$ vertices
(not including $a$ and $b$). From the context it will be clear which of the two half-planes we mean.

The following theorem 
gives an upper bound on the degeneracy of 
outer $k$-planar graphs.

\begin{theorem}
  \label{thm:maxmindeg}
  For every positive integer~$k$, let $\delta_k$ be the degeneracy
  of outer $k$-planar graphs.  
  Then~${\delta_k \le \lfloor c_k\sqrt{k} \rfloor }$, where
  \[
%    c_k = \frac{5}{4\sqrt{k}} + \frac 34\sqrt{\frac 1k +8}.
    c_k = \frac{5 + 3\sqrt{1+8k}}{4\sqrt{k}}.
  \]
  The sequence $(c_k)_{k\ge1}$ is  monotonically 
  decreasing with $c_1
  = 3.5$ and limit $3\sqrt{2}/2$.
\end{theorem}
\begin{proof}
Let $G$ be an outer $k$-planar graph of minimum degree~$\delta_k$.
Consider some outer $k$-planar drawing of $G$.
Assume that there exists an edge~$e$ that
splits off $t\in \mathbb{N}$ vertices in the drawing of~$G$ to one side, 
then there are at least~${\delta_k t -  t(t-1) - 2t = 
\delta_k t - t(t+1)}$ edges crossing the edge~$e$ (on the left-hand side of the equality the second term 
stands for the sum of the degrees of a clique on $t$ vertices and the third term for the number of edges incident to the endpoints of $e$ and to the $t$ many vertices). 
Because~$G$ is outer $k$-planar, we have that
\begin{equation}
  \label{eq:quadratic}
  \delta_k t - t(t+1) \le k.
\end{equation}
Therefore, either $t \le t_1$ or $t_2 \le t$, where
\[t_1 = \left((\delta_k-1)-\sqrt{(\delta_k - 1)^2 - 4k}\right)/2
  \quad\text{and}\quad
  t_2 = \left((\delta_k-1)+\sqrt{(\delta_k - 1)^2 - 4k}\right)/2.\]
Assume for a contradiction that $\delta_k \ge c\sqrt{k}$ for some
$c>c_k$, where $c_k$ is as stated in the theorem.
% \[
% %  c_k = \frac{5}{4\sqrt{k}} + \frac 34\sqrt{\frac 1k +8}.
%   c_k = \frac{5 + 3\sqrt{1+8k}}{4\sqrt{k}}.
% \]
Then the solutions $t_1$ and $t_2$ ($t_1<t_2$) to the quadratic 
equation~\eqref{eq:quadratic} exist and are
distinct, because $c_k > 1/\sqrt{k} + 2$,
and so, $\delta_k > c_k \sqrt{k} \ge 2\sqrt{k} + 1$ and the
discriminant of~\eqref{eq:quadratic} is positive.
Call an edge that splits off 
at least~$t_2$ vertices to both sides \emph{long}. 
The number~$\ell$ of long edges incident to each vertex
can be bounded from below by subtracting from $\delta^*$ the upper
bound on the number of incident edges that are not long. 
Recall that those can split off at most $t_1$ vertices, so we can have
at most $2(t_1+1)$ such edges incident to the same vertex. 
Thus we can bound $\ell$ from below as follows:
\begin{align*}
  \ell &\ge \delta_k - 2(t_1+1) \\
       &= \delta_k - 2\left(\left((\delta_k-1)-\sqrt{(\delta_k-1)^2
         - 4k}\right)/2 + 1\right) \\
       &=   1 + \sqrt{(\delta_k - 1)^2 - 4k} - 2
         \ge \sqrt{(c^2-4)k - 2c\sqrt{k}+1} - 1. 
\end{align*}
Take a shortest long edge~$e'$ in the outer $k$-planar drawing of $G$, that is, in one of the open half-planes $h$ defined by~$e'$, there is no long edge 
that is completely contained in $h$. 
Let $V_h$ be the set of vertices of the graph that are contained in the half-plane $h$. 
Because~$e'$ is a long edge, $|V_h| \ge t_2 = \left((\delta_k-1)+\sqrt{(\delta_k - 1)^2 - 4k}\right)/2$. 
Because it is a shortest long edge,
all the long edges
incident to the vertices in $V_h$ must cross $e'$.  Therefore, 
the number of long edges that cross~$e'$ is at least 
\begin{align*}
\ell t_2 \ge  &\frac{1}{2}\left(\sqrt{(c^2-4)k - 2c\sqrt{k}+1} - 1\right)\left((c\sqrt{k}-1)+\sqrt{(c^2-4)k - 2c\sqrt{k}+1}\right)
\intertext{Let}
f(c, k)  = &\frac{1}{2}\left(\sqrt{(c^2-4)k - 2c\sqrt{k}+1} - 1\right)\left((c\sqrt{k}-1)+\sqrt{(c^2-4)k - 2c\sqrt{k}+1}\right) - k.
\end{align*}
Consider the equation $f(c, k) = 0$, where $k \in \mathbb{N}$
and $c \in (1/\sqrt{k} + 2, \infty)$.

In order to find the roots of the above equation, we
first observe that
\[(c^2-4)k - 2c\sqrt{k}+1 = (c\sqrt{k}-1)^2 - 4k.\]
Therefore, $f(c, k)  =$
\[
  \frac{1}{2}\left(
    (c\sqrt{k}-1)^2 - 4k  + (c\sqrt{k}-1)\sqrt{(c\sqrt{k}-1)^2 - 4k} -
    \sqrt{(c\sqrt{k}-1)^2 - 4k} - (c\sqrt{k}-1)
  \right) - k.
\]
Let $X = c\sqrt{k}-1$, then the equation $f(c, k) = 0$ becomes
\begin{equation*}
  X^2 - 4k + X\sqrt{X^2 -4k} - \sqrt{X^2 -4k} - X - 2k = 0
\end{equation*}
or 
\begin{equation*}
  X^2 -X -6k =  (1-X) \sqrt{X^2 -4k}
\end{equation*}
let us square it
\begin{equation*}
  X^4 + X^2 + 36k^2  - 2X^3  - 12X^2k + 12 Xk = 
  (1-x)^2(X^2 - 4k)
\end{equation*}
after expanding we have
\begin{equation*}
   -8X^2k +4Xk +36k^2 + 4k = 0
\end{equation*}
or
\begin{equation*}
   X^2 - \frac 12X - \frac{9k+1}{2} = 0
\end{equation*}
the positive root 
\begin{equation*}
   X = \frac{1+3\sqrt{1+8k}}{4}.
\end{equation*}
Then, due to our definition of $X$, we have
\begin{equation*}
  c = \frac{X+1}{\sqrt{k}} = \frac{5 + 3\sqrt{1+8k}}{4\sqrt{k}}.
  % = \frac{5}{4\sqrt{k}} + \frac{3}{4} \sqrt{\frac{1}{k} + 8}.
\end{equation*}

Note that the right-hand side is actually $c_k$.  This contradicts our
assumption that $c>c_k$.  Thus, for every $k \in \mathbb{N}$, the
degeneracy of any outer $k$-planar graph is at most
$c_k\sqrt{k}$.  Observe that $(c_k)_{k\ge1}$ is a monotonically
decreasing sequence with $c_1 = 3.5$ and limit $3\sqrt{2}/2$.
\end{proof}
Lemma~\ref{obs:largest_clique} implies that there are (complete) outer
$k$-planar graphs with degeneracy $\lfloor 2\sqrt{k} \rfloor +1$.
Note that, for $k<42$, our upper bound
$\lfloor (5 + 3\sqrt{1+8k})/4 \rfloor \approx 2.13 \sqrt{k}$ from
Theorem~\ref{thm:maxmindeg} differs only by~1 from the existential
lower bound provided by Lemma~\ref{obs:largest_clique}.

As a direct consequence of Theorem~\ref{thm:maxmindeg}, we obtain the
following.

\begin{corollary}
\label{cor:edge-density}
Every outer $k$-planar graph has at 
most $\lfloor c_k\sqrt{k} \rfloor n$ edges, where $n$ is the number of vertices. 
\end{corollary}

For small $k$ (recall that $c_1 = 3.5$) there is a stronger (and more
general) bound by Aichholzer, Obenaus, Orthaber, Paul, Schnider,
Steiner, Taubner, and Vogtenhuber~\cite{aoopsstv-epocgg-SoCG22} of
roughly $2.46\sqrt{k}n$ for the maximum edge density of outer
$k$-planar graphs.  However, as $k$ tends to infinity, our upper bound
$(5 + 3\sqrt{1+8k})/4 \cdot n$ tends to
$(3\sqrt{2}/2)\sqrt{k}n \approx 2.13\sqrt{k}n$.  Hence, our bound is
less than the bound by Aichholzer et al.\ for $k \ge 15$ (ignoring the
floors).

By combining Lemma~\ref{obs:largest_clique} and Theorem~\ref{thm:maxmindeg}, we obtain the following. 
\begin{corollary}
\label{thm:coloring}
 Every outer $k$-planar graph can be colored with 
 $\lfloor c_k\sqrt{k}
\rfloor +1$ colors. 
 There exist outer $k$-planar graphs that need at least $\lfloor\sqrt{4k+1}\rfloor + 2$ colors.
\end{corollary}

\subsection{Quasi-Polynomial-Time Recognition via Balanced Separators}
\label{sec:separators}

We now show that outer $k$-planar graphs have separation number at
most $2k+3$~(Theorem~\ref{thm:separator}).  Via a result of
Dvo\v{r}\'ak and Norin~\cite{dn-tgbs-JCTS19}, this implies that their
treewidth is at most $30k+45$.  However, a result of Wood and
Telle~\cite[Proposition~8.5]{Wood2007} implies that every outer
$k$-planar graph has treewidth at most $3k+11$, which is a better
bound than what we get by applying the result of Dvo\v{r}\'ak and
Norin to our separators.  The treewidth bound of $3k+11$ in turn
implies a separation number of $3k+12$, but our bound is better.
Our separators also allow us to test outer $k$-planarity in
quasi-polynomial time; see Theorem~\ref{thm:recognition-algorithm}.

Very recently, both the result of Wood and Telle and our result have
been improved by Firman, Gutowski, Kryven, Okada, and
Wolff~\cite{fgkow-btokp-GD24}, who showed that outer $k$-planar graphs
have treewidth at most $\lfloor 1.5k \rfloor + 2$ and separation
number at most $k+2$.
They also show that $k+2$ is an existential lower bound for both
numbers.

\begin{theorem}
  \label{thm:separator}
  Each outer $k$-planar graph has separation number at most $2k+3$.
\end{theorem}

\begin{proof}
Consider an outer $k$-planar drawing.
If the graph contains an edge~$e$ that splits off at least
$n/3$ and at most $2n/3$
vertices to one side, we can use~$e$ to obtain a balanced
separator of size at most $k+2$: take the endpoints of~$e$ 
and a vertex cover of the edges that cross~$e$.
Thus, in the following we assume that no such edge exists.

Consider a pair of vertices $a$ and $b$ such that the \emph{line} through $ab$  divides
the drawing into left and right sides having an almost equal number of
vertices (with a difference of at most one).
If the edges which cross the line $ab$ also mutually cross each other,
there can be at most $k$ of them.
Thus, we again have a balanced separator of size at most $k+2$.
So assume now that there is a pair of edges
that cross the line~$ab$, but do not cross each other. We call such a pair of edges \emph{parallel}.
We now pick a pair of parallel edges in a specific way.
Starting from~$b$, let $b'$ be the first vertex along 
the boundary in clockwise direction
such that there is an edge $b'b''$ that crosses the line $ab$.
Symmetrically, starting from~$a$, let $a'$ be the first vertex along the
boundary in clockwise direction
such that there is an edge $a'a''$ that crosses the line $ab$;
see Fig.~\ref{fig:aabb}.
Note that the edges $a'a''$ and~$b'b''$ are either identical or 
parallel.
In the former case, we see that all other edges crossing the line $ab$ 
must also cross the 
edge $a'a'' = b'b''$, and hence there are at most $k$ 
edges crossing the line~$ab$.
In the latter case, there are two subcases that we treat below.

For two vertices $u$ and $v$, let \emph{$[u,v]$} be the set of
vertices that starts with $u$ and, going clockwise, ends with~$v$.
Let \emph{$(u,v)$} $=[u,v] \setminus \{u, v\}$.  When we say that an
edge~$uv$ splits off $\alpha$ vertices without specifying the
side, we mean that $|[u,v]| = \alpha$.  Note that this
implies that the edge~$vu$ splits off $n-\alpha+2$ vertices.

\begin{figure}[htb]
    \begin{subfigure}[b]{0.33\textwidth}
      \centering
      \includegraphics[page=1]{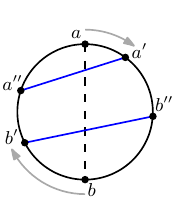}
      \caption{edges $a'a''$ and
      $b'b''$ are parallel}
      \label{fig:aabb}
    \end{subfigure}
    \hfill
    \begin{subfigure}[b]{0.26\textwidth}
      \centering
      \includegraphics[page=2]{okptw}
      \caption{case~1}
      \label{fig:case1}
    \end{subfigure}
    \hfill
    \begin{subfigure}[b]{0.3\textwidth}
      \centering
      \includegraphics[page=3]{okptw}
      \caption{case~2}
      \label{fig:case2}
    \end{subfigure}

    \caption{Illustration for the proof of Theorem~\ref{thm:separator}}
    \label{fig:bs}
\end{figure}

Because we assumed that there is no edge that splits off at least
$n/3$ and at most $2n/3$ vertices, our case distinction below is
exhaustive.

\smallskip
\noindent{\bf Case 1.} {\it The edge $b'b''$ splits off
  at most $\mu \le n/3$ vertices (see Fig.~\ref{fig:case1}).}

\noindent
In this case, $n/3 \le |[b,b']| \le n/2$ or
$n/3 \le |[b'',b]| \le n/2$.  We claim that neither the line $bb'$ nor
the line $bb''$ is crossed more than $k$ times.
Namely, due to the choice of~$b'$, each edge that crosses
the line~$bb'$ also crosses the edge~$b'b''$.
Also due to the choice of~$b'$, each edge that crosses the
line~$bb''$ also crosses the edge~$b'b''$.
Thus, we have a separator of size at most $k+2$, regardless of whether we
choose $bb'$ or $bb''$ to separate the graph.  As we observed above,
one of them is balanced.

\smallskip
\noindent{\bf Case 1$'$.} {\it The edge $a'a''$ splits off at most $
  n/3$ vertices.} \\
This is symmetric to case~1.

\smallskip
\noindent{\bf Case 2.} {\it The edge $a'a''$ and the edge $b'b''$
  both split off at most $ n/3$ vertices (see Fig.~\ref{fig:case2}).}

\noindent We show that we can even find such a pair of edges 
$a'a''$ and $b'b''$ with the additional requirement that there is no
edge $cd$ between them that is parallel to them, that is,
$c \in [a',b'']$ and $d \in [b',a'']$ and $cd \not\in \{a'a'',b'b''\}$.
We call such a pair \emph{close}.
If there is an edge $e = uv$ between $a''a'$ and
$b'b''$, we form a new pair by using
$uv$ and $a'a''$ if $uv$ splits off at most $n/3$ vertices or by using
$uv$ and $b'b''$ if $vu$ splits off at most $n/3$ vertices.  Note that
one of the two conditions must hold due to our assumption that there
is no edge that splits off at least $n/3$ and at most $2n/3$ vertices.
By repeating this procedure, we always find a close pair.
Hence, we can assume that $a'a''$ and $b'b''$ actually form a close pair.
Let $\alpha=|(a'',a')|$, $\beta=|(b'', b')|$, $\gamma=|(a',b'')|$, 
and $\delta=|(b',a'')|$; see Fig.~\ref{fig:case2}.
Note that $n-4 \le \alpha+\beta+\gamma+\delta \le n-3$.

We further consider two cases depending on whether 
the two edges $a''a'$ and $b''b'$ share an endpoint or not.

First, suppose that $a' = b''$ or $a'' = b'$
(that is, $\gamma=0$ or $\delta=0$).
We can now use both edges $a''a'$ and
$b''b'$ (together with any edges crossing them) to obtain a separator
of size at most $2k+3$.  The separator is balanced since $\alpha+\beta
\le {2n}/{3}$ and $\gamma+\delta \le {n}/{2}$.
The latter holds due to the fact that $\gamma=0$ or $\delta=0$
and that $\gamma \le n/2$ and $\delta \le n/2$ since, on each side
of the line~$ab$, there are at most $n/2$ vertices.

Second, assume that $a'$, $a''$, $b'$, $b''$ are all distinct.
Note that we again have that $\gamma \le n/2$ and $\delta \le n/2$.
% since each side of the line $ab$ has at most $n/2$ vertices.
We separate the graph along the line $a'b'$; see Fig.~\ref{fig:case2}.
Namely, all the edges that cross this line
must also cross the edge~$a''a'$ or the edge~$b''b'$.
Therefore, we obtain a separator of size at most $2k+2$.

To see that the separator is balanced, we consider two cases.
If $\delta \geq  {n}/{3}$ (or $\gamma \geq  {n}/{3}$), 
then $\alpha+\beta+\gamma \le {2n}/{3}$ (or $\alpha+\beta+\delta 
\le {2n}/{3}$).
Otherwise $\delta< n/3$ and $\gamma < {n}/{3}$.
In this case $\delta+\alpha \le {2n}/{3}$ and $\gamma+\beta \le 
{2n}/{3}$.  In both cases the separator is balanced.
\end{proof}

Let $G$ be an outer $k$-planar graph and consider some
outer $k$-planar drawing $D$ of $G$.  According to
Theorem~\ref{thm:separator}, $G$ has a
balanced separator of size at most $2k+3$ 
that arises from $D$ as described in Theorem~\ref{thm:separator}.
Furthermore, we observe the following structure of the separator that
follows from the proof of Theorem~\ref{thm:separator}.

\begin{observation}
  \label{obs:separator-types}
  Each outer $k$-planar graph with an outer $k$-planar drawing~$D$ has
  a balanced separator of size at most $2k+3$ that is of one of the
  following two types:
  \begin{enumerate}[(1)]
  \item three vertices $x$, $y$, and $z$ that are endpoints of two
    close parallel edges~$xy$ and~$yz$ and some endpoints of the
    edges crossing the edges~$xy$ and~$yz$ in~$D$ and
  \item two vertices $x$ and $y$
    % {x,y} = {a,b} or {x,y} = {a',b'}
    and some endpoints of the edges crossing the line~$xy$ in~$D$.
  \end{enumerate}
\end{observation}

We call the vertices $x, y$, and $z$ \emph{boundary vertices}. 

Next we show that we can recognize
whether a given graph $G$ is
outer $k$-planar in quasi-polynomial time; 
see Theorem~\ref{thm:recognition-algorithm}.
Our proof relies on Theorem~\ref{thm:separator}; in particular, 
it uses the structure of balanced separators
as described in Observation~\ref{obs:separator-types}. 

We need the following notation.  Given an outer $k$-planar drawing~$D$
of~$G$, let~$S$ be a separator constructed as in the proof in
Theorem~\ref{thm:separator}, and let $S'\subset S$ be the set of
boundary vertices.
% Note that $S$ is balanced and consists of at most $4k+3$ % 2k+3?!
% vertices.  Consider the subdrawing of~$D$ induced by~$S$.
We call each arc of the circle of 
the drawing~$D$ that connects two consecutive 
vertices of~$S'$ a \emph{region} with respect to~$S'$.  We 
call the vertices in $S\setminus S'$ \emph{regional} vertices.  Note
that if the separator is of type~(1), then there are three regions, and
if it is of type~(2), then there are two regions.

\begin{lemma}
  \label{lem:components}
  For each connected component of $G\setminus S$, its vertices are
  only in one region of~$D$ with respect to~$S'$.
\end{lemma}
\begin{proof}
  Assume for a contradiction that there is a connected component with
  an edge $e$ with endpoints in two different regions.  Then $e$
  crosses a line between two boundary vertices in~$D$.  Therefore, $e$
  is a separator edge and at least one of its endpoints lies in~$S$.
\end{proof}
We also have the following.
\begin{lemma}
  \label{lem:three-vertices}
  If $S$ is a type-(1) separator, then no connected component of
  $G\setminus S$ is adjacent to all three boundary vertices.
\end{lemma}
\begin{proof}
  Assume for a contradiction that there is a connected component of
  $G\setminus S$ connected by edges in
  $E(G) \setminus E(G\setminus S)$ to each of the boundary
  vertices~$x$, $y$, and~$z$.
  If the component lies in the region bounded by $x$ and $z$,
  then the edges $xy$ and $yz$ are not a close parallel pair in~$D$.
  This yields the desired contradiction.  If the component lies in one
  of the regions bounded by~$x$ and~$y$ or by~$y$ and~$z$, an edge~$e$
  of the component crosses either the line~$xy$ or the line~$yz$
  in~$D$.  Again, this yields the desired contradiction since then $e$
  is one of the separator edges.
\end{proof}

In the remainder of this section, we refer to the connected components
of $G\setminus S$ simply as \emph{components}.

In order to obtain a quasi-polynomial algorithm for recognizing outer
$k$-planar graphs, we leverage the structure of balanced separators 
as described in the proof of Theorem~\ref{thm:separator}.
Specifically, we enumerate the vertex sets that could form such a
separator.  For each set, we choose an appropriate outer $k$-planar
drawing of the subgraph it induces
and partition the remaining parts of the graph into 
regions bounded by the boundary vertices of the separator.
We then recursively test whether the graph 
induced by the vertices in each of the 
regions admits an outer $k$-planar drawing.

\begin{theorem}
  \label{thm:recognition-algorithm}
  For every fixed $k$, testing whether a given graph is outer
  $k$-planar takes $O(2^{\mathrm{polylog}(n)})$ time, where $n$ is the
  number of vertices of the given graph.
\end{theorem}

\begin{proof}
  To obtain quasi-polynomial runtime, we need to limit the
  number of regions on which we branch.
Let $G$ be the given graph.
According to Theorem~\ref{thm:separator}, 
if $G$ is outer $k$-planar, 
then there exists a set $S$
with at most $4k+3$ vertices 
where the vertices are endpoints of the
edges of a balanced separator. 
Note that we are interested in all the endpoints of the edges, not just those that form the separator.
In addition, note that $S$ is also a balanced separator. 
Furthermore, $S$ is of one of the two types described
in Observation~\ref{obs:separator-types}; see
Fig.~\ref{fig:separator-cases-components-type-1} for a type-(1)
separator and Fig.~\ref{fig:separator-cases-components-type-2} for a
type-(2) separator.
We find $S$ by brute force, enumerating all vertex sets of size at most $4k + 3$.  
Once $S$ is fixed, observe that the subgraph~$G[S]$ induced by~$S$ 
admits only $f(k)$ possible outer $k$-planar drawings, for some function $f$.  
We then enumerate all such possible drawings.
Assume now that one such drawing of~$S$ is fixed.

\begin{figure}[bt]

  \hfill
  \begin{subfigure}[b]{0.30\textwidth}
    \centering
    \includegraphics[page=1]{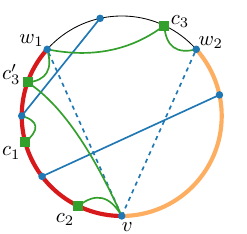}
    \caption{a separator of type (1)}
    \label{fig:separator-cases-components-type-1}
  \end{subfigure}
  \hfill
  \begin{subfigure}[b]{0.30\textwidth}
    \centering
    \includegraphics[page=2]{separator-cases-components}
    \caption{a separator of type (2)}
    \label{fig:separator-cases-components-type-2}
  \end{subfigure}
  \hfill\null
  
  \caption{The two types of separators according to
    Observation~\ref{obs:separator-types}.}
  \label{fig:separator-cases-components}
\end{figure}

We now consider two cases depending on the type of the separator $S$.
If $S$ is of type~(1), then 
$S$ contains three boundary vertices~$v$, $w_1$,
and~$w_2$; see Fig.~\ref{fig:separator-cases-components-type-1}.
If $S$ is of type~(2), then 
there are two boundary vertices $v$ and
$w$; see Fig.~\ref{fig:separator-cases-components-type-2}.
Note that, since we have a fixed drawing of~$G[S]$, the regional vertices
are partitioned into regions as defined by the boundary vertices of~$S$.
According to Lemma~\ref{lem:components}, no component of
$G \setminus S$ is connected to regional vertices of different
regions.

\subparagraph*{Separator of type~(1)}

First, assume that the separator~$S$ is of type~(1);
see Fig.~\ref{fig:separator-cases-components-type-1}.
We start by choosing the three boundary vertices~$v$, $w_1$,
and~$w_2$ from~$S$.
According to Lemma~\ref{lem:components}, if there is a component connected to regional vertices
of different regions, we can reject this
configuration.
Furthermore, according to Lemma~\ref{lem:three-vertices},
no component is adjacent to all three boundary vertices.
We now consider three possible different types~$c_1$,
$c_2$, and~$c_3$ of components depending on the type of the vertices
 (boundary or regional) that they are connected to; 
 see Fig.~\ref{fig:separator-cases-components-type-1}.

%  that can occur in a region \emph{neighboring} $w_1$ 
%  (that is, one of the boundary vertices of the region is $w_1$).
%  A symmetric argument holds for the region neighboring $w_2$.

 Components of type $c_1$ are connected to (possibly many) 
 regional vertices of the same region and may be connected to boundary vertices as well.
 In any valid drawing, they will end up in the same region as their regional vertices.

Components of type $c_2$ are not connected to any regional vertices;
instead, they are connected to exactly one of the boundary vertices.
 Since they are not connected to regional vertices, 
 they cannot interfere with other parts of the drawing, 
 so we can arbitrarily assign them to a region adjacent
 to the boundary vertex they are connected to.

Finally, there are components that are connected to two boundary vertices.
There are two types of such components.

The first type (which we call~$c_3$) 
consists of those components that are connected to~$w_1$ and~$w_2$.
These have to be placed in the region bounded by~$w_1$ and~$w_2$.
Otherwise one of the edges connecting such a component would cross the
line~$vw_1$ or the line~$vw_2$, and thus, it would be a separator edge
and therefore it could not be part of the component.

The second type (which we call~$c_3'$) consists of those components
that are either connected to~$v$ and~$w_1$ or to~$v$ and~$w_2$.
 These have to be placed in the region
 bounded by~$v$ and~$w_1$ or by~$v$ and~$w_2$, respectively, since
 otherwise (similarly as above)
 one of the edges connecting such a component would cross the
 line~$vw_1$ or the line~$vw_2$ and thus, it would be a separator
 edge and could not be part of the component.
%  And components of type $c_4$ that are connected to $v$ and $w_1$ or $v$ and $w_2$
%  and are placed in the region bounded $w_1$ and $w_2$.
%  %
Note that a component connected to~$v$ and~$w_1$ or to~$v$ and~$w_2$
 cannot be placed in the region between $w_1$ and $w_2$
since this would contradict the fact that
our separator arose from a pair of close parallel edges as argued 
in the proof of Theorem~\ref{thm:separator}. 
 
The above discussion implies the following for our recognition
algorithm.  If, in a fixed configuration (i.e., set~$S$, drawing
of~$G[S]$, and triplet of boundary vertices), the drawing of~$G[S]$
corresponds to a separator of type~(1), then we can reject the
current configuration (based on having components that do not
correspond to any of the types~$c_1$, $c_2$, $c_3$ or~$c_3'$)~--
or every component of $G \setminus S$ is either attached to exactly
one boundary vertex or it has a well-defined placement into a region
defined by its adjacent boundary vertices.  For each component that is
attached to exactly one boundary vertex, it suffices to recursively
produce a drawing of that component together with its boundary vertex
and to place this drawing next to the boundary vertex.
We place the other components into their regions and recurse on the regions.
This covers all cases for separators of type~(1).

\subparagraph*{Separator of type~(2)}

Second, assume that the separator $S$ is of type~(2); see
 Fig.~\ref{fig:separator-cases-components-type-2}.
Note that we now have two boundary vertices $v$ and $w$;
thus we have only two regions.
Similarly as above, we now consider three possible different
types~$c_1$, $c_2$, and~$c_4$ of components depending on the type of
vertices (boundary or regional) that they are connected to; 
see Fig.~\ref{fig:separator-cases-components-type-2}.

Components $c_1$ and $c_2$ are defined exactly as they were defined
for separators of type~(1).
% that is, components of type $c_1$ are connected to (possibly many)
% regional vertices of the same region and may be connected to
% boundary vertices as well, and components of type $c_2$ are not
% connected to any regional vertices; instead, they are connected to
% exactly one boundary vertex.
We handle these components as described above.

A separator of type~(2) may, however, also create components that are
connected to both boundary vertices~$v$ and~$w$, but to no regional
vertices; we refer to these as components of type~$c_4$.
Such components can be placed in either of the two regions.

If $G[S]$ contains an edge $v'w'$ that crosses the line~$vw$ (as in
Fig.~\ref{fig:separator-cases-components-type-2}), then there cannot
be more than $k$ components of type~$c_4$.  Namely, in any drawing,
every type-$c_4$ component contains an edge that connects
the component to either~$v$ or~$w$ and that hence crosses~$v'w'$.
Thus, we enumerate all the different placements of these components
and recurse accordingly.

On the other hand, if $G[S]$ does not contain any edge that crosses
the line~$vw$, then the separator is exactly the pair $\{v,w\}$. 
In this case, there are no components of type~$c_1$.
Components of type~$c_2$ can be handled as before.
We now argue that we can have at most a function of $k$ many
different components of type~$c_4$ in a valid drawing.
 Consider the components of type~$c_4$.
In a valid drawing, each type-$c_4$ component defines an interval
of the region from its \emph{highest} vertex (that is, the vertex
closest to~$w$) that is adjacent to~$v$ or~$w$ to its \emph{lowest}
vertex (that is, the vertex closest to~$v$) that is adjacent to~$v$
or~$w$.
 Two such intervals relate in one of three ways:
 they overlap each other, they are disjoint, or one contains the other.
 We group components with either overlapping or disjoint intervals into \emph{layers}.
We depict this situation in Fig.~\ref{fig:separator-cases-components-c}
where, for simplicity, we draw only the highest and the lowest vertex
of each component, and we represent the path between them by a single
edge in the drawing. 

\begin{figure}[tb]
  \centering
  \includegraphics{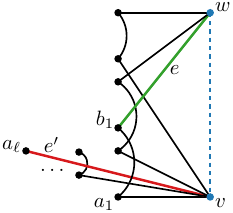}
  \caption{If the pair $\{v,w\}$ is a separator, we group the
    components of $G-\{v,w\}$ into layers.}
  \label{fig:separator-cases-components-c}
\end{figure}

Let $a_1b_1$ be the interval of the bottommost type-$c_4$ component
(i.e., going from~$v$ in clockwise direction, $a_1$ is the first
vertex adjacent to~$v$ or~$w$ in a type-$c_4$ component).
The first layer consists of this component and every component
whose interval either overlaps or is disjoint from the interval~$a_1b_1$.
We now consider the edge~$e \in \{b_1w, b_1v\}$ (whichever exists;
drawn green in Fig.~\ref{fig:separator-cases-components-c}).
For every component of the first layer whose interval is disjoint from
$a_1b_1$, the edge that connects the component to~$v$ crosses~$e$.
For every component of the first layer whose interval overlaps
$a_1b_1$, the edge~$e$ is crossed by at least one edge within that
component.  Hence, the first layer consists of at most $k+1$
components.  The next layer consists of all components whose 
intervals are contained in the intervals of components on the first
level etc.
Let~$\ell$ be the index of the deepest layer, let~$a_\ell$ be the
bottommost vertex of that layer connected to~$v$ or~$w$.
If~$e'$ is the edge that connects the component of~$a_\ell$ to~$v$
(red in Fig.~\ref{fig:separator-cases-components-c}), then~$e'$
is crossed by some edge of every layer above it.  Hence there are at
most $k+1$ levels and at most $(k+1)^2$ components per region.  This
means that we can enumerate their possible placements and recurse
accordingly.

We now bound the runtime of the above recursive algorithm.
Let $T(n)$ denote its worst-case runtime when applied to an outer
$k$-planar graph with $n$ vertices.  Then,
\[
  T(n) = \begin{cases}
          n^{O(k)} \cdot f(4k+3) \cdot g((k+1)^2) \cdot n^4 \cdot 
T({2n}/{3}) & \text{ for } n > 5k, \\
          f(n) & \text{ otherwise,}\\
         \end{cases}
\]
where $f(s)$ denotes the number of different outer $k$-planar drawings
of a graph with $s$ vertices and $g(s)$ denotes the number of ways $s$
components can be partitioned into at most three regions.  The factor
$n^{O(k)}$ stands for finding all possible separators of size $4k+3$,
and $T({2n}/{3})$ is the runtime of the recursive call on a region.

Thus, the algorithm runs in quasi-polynomial time, i.e., $T\in 2^{\mathrm{polylog}(n)}$.
\end{proof}

\section{Outer Quasi-Planar Graphs}
\label{sec:o-kqp}

In this section we consider outer quasi-planar graphs, that is, outer
3-quasi-planar graphs.  We describe some classes of graphs that are
outer quasi-planar and some classes of graphs that are not outer
quasi-planar.  In particular, we
show that there are planar graphs that are not outer quasi-planar.
Hence, the classes of planar graphs and outer quasi-planar graphs are
incomparable; see Theorem~\ref{thm:containment}.

A graph is \emph{sub-Hamiltonian} if it is a subgraph of a Hamiltonian
graph.  A \emph{planar $3$-tree} is either $K_4$ or a graph~$G$ formed
by \emph{stacking} a vertex into a smaller planar $3$-tree~$G'$, that
is, by inserting a vertex into a triangular face of~$G'$ and by
connecting the new vertex to the three vertices of the triangular face
of~$G'$.  The graph~$K_4$ is the only complete planar $3$-tree of
level~1.  For $i>1$, the \emph{complete planar $3$-tree of level~$i$}
is obtained from the complete planar $3$-tree of level~$i-1$ by
stacking a vertex into every triangular face of the complete tree of
level $i-1$.

\begin{proposition}
\label{prop:someoqp}
The following graphs are outer quasi-planar:
(a)~every sub-Hamiltonian planar graph;
(b)~$K_{p,q}$ for $p \le 4$ and $q \le 4$;
(c)~$K_{n}$ for $n \le 5$;
(d)~the complete planar $3$-tree of level at most~3;
(e)~grid graphs of any size.
\end{proposition}

\begin{proof}
  Clearly, every subgraph of an outer quasi-planar graph is outer
  quasi-planar.

  To see~(a), consider a planar drawing of the given graph~$G$.  It
  partitions the edge set of~$G$ into the set~$E_0$ edges that lie on
  the cycle, the set~$E_1$ of edges that lie in the interior of the
  cycle, and the set~$E_2$ of edges that lie in the exterior of the
  cycle.  Now place the vertices of~$G$ on a circle so that their
  order follows the Hamiltonian cycle.  The induced straight-line
  drawing of~$G$ is outer quasi-planar since the edges in~$E_0$ do not
  cross any edges, no two edges in~$E_1$ cross, and no two edges
  in~$E_2$ cross.  Hence, every set of three edges can have at most
  two crossings.

  For~(b) and~(c), note that the two drawings of~$K_5$ and~$K_{4,4}$
  in Fig.~\ref{fig:k5andk44} are outer quasi-planar.  (For $K_5$
  this is trivial, because $K_5$ does not contain three independent
  edges.)

\begin{figure}[htb]
    \hfill
    \begin{subfigure}{0.2\textwidth}
      \centering
      \includegraphics{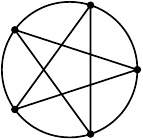}
      \caption{$K_5$}
    \end{subfigure}
    \hfill
    \begin{subfigure}{0.2\textwidth}
      \centering
      \includegraphics{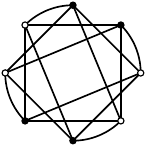}
      \caption{$K_{4,4}$}
    \end{subfigure}
    \hfill\null

    \caption{Drawings showing that $K_5$ and $K_{4,4}$ are outer
      quasi-planar.}
    \label{fig:k5andk44}
\end{figure}

  We have verified claim~(d) by constructing a SAT formulation; see
  Appendix~\ref{sec:oqpchecker}.

  Claim~(e) follows from the fact that
  grid graphs are sub-Hamiltonian.  Note that grid graphs of odd order
  are not Hamiltonian, but they are subgraphs of larger grids of even
  order, which are Hamiltonian. Therefore, all grids are
  sub-Hamiltonian.
\end{proof}

Below, we identify complete and complete bipartite graphs that are not
outer quasi-planar.  Furthermore, not all planar graphs are
outer quasi-planar, e.g., Fig.~\ref{fig:3tree}(a) shows a
planar 3-tree that is not outer quasi-planar (but removing any vertex
renders it outer quasi-planar).  This was verified using a SAT
formulation; see Appendix~\ref{sec:oqpchecker}.
The drawing of the graph in Fig.~\ref{fig:3tree}(b) was constructed
by removing the blue vertex and drawing the remaining graph in an outer
quasi-planar way.

\begin{proposition}
\label{prop:somenotoqp}
The following graphs are not outer quasi-planar:
(a)~$K_{p,q}$, for $p\ge3$ and $q \ge 5$;
(b)~$K_{n}$, for $n\ge6$;
(c)~every planar 3-tree with at least four complete levels.
\end{proposition}

\begin{proof}
  Using the SAT formulation in Appendix~\ref{sec:oqpchecker}, we
  verified that $K_{3,5}$, $K_6$, and the planar 3-tree with four
  complete levels are all not outer quasi-planar.  Clearly, every
  graph that contains any of these three graphs as subgraphs is
  also not outer quasi-planar.
\end{proof}

Together, Propositions~\ref{prop:someoqp} and~\ref{prop:somenotoqp}
immediately yield the following.
\begin{theorem}
  \label{thm:containment}
  The class of planar graphs and the class of outer quasi-planar
  graphs are incomparable under containment.
\end{theorem}
\begin{remark}
For outer $k$-quasi-planar graphs with $k > 3$, containment questions
become more intricate.  Every planar graph is outer 5-quasi-planar because
planar graphs have page number~4 \cite{YANNAKAKIS198936}.  There are
planar graphs that are not outer quasi-planar (the planar 3-trees
with at least four complete levels).  It is open
whether every planar graph is outer 4-quasi-planar.
\end{remark}

\begin{figure}[h]
    \begin{subfigure}[b]{0.46\textwidth}
      {\small\sffamily\bfseries (a)}\hspace{-3ex}\includegraphics[page=1]{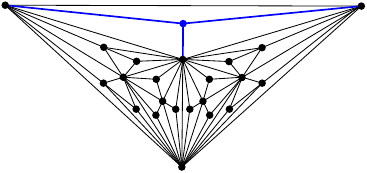}
    \end{subfigure}
    \hfill
    \begin{subfigure}[b]{0.3\textwidth}
      {\small\sffamily\bfseries (b)}\hspace{-3ex}\includegraphics[page=2]{big}
    \end{subfigure}
    \caption{A planar 3-tree $G$ (with 23 vertices) that is not outer
      quasi-planar:
      % (but removing any vertex makes it outer quasi-planar):
      (a)~planar drawing of~$G$; (b)~convex drawing of~$G$ that is not
      outer quasi-planar (see the highlighted triangular inner faces)
      but contains an outer quasi-planar drawing of $G$ minus the blue
      vertex.}
    \label{fig:3tree}
\end{figure}

\section{Testing for Full Convex Drawings via MSO$_2$}
\label{sec:full}

The class of \emph{full outer $k$-planar graphs} was introduced by 
Hong and Nagamochi~\cite{hn-ltatfo2p-DAM19}.
Recall that this class consists of the graphs that admit a $k$-planar
convex drawing where no crossing lies on the boundary of the outer face.
Hong and Nagamochi gave a linear-time
recognition algorithm for full outer $2$-planar graphs.
They state that a graph $G$ is (full) outer $2$-planar if and only 
if its biconnected components are (full) outer $2$-planar and that 
the outer boundary of a full outer $2$-planar drawing of a 
biconnected graph $G$ is a Hamiltonian cycle of $G$.
We call the subclasses of outer $k$-planar and outer $k$-quasi-planar graphs   
that have a convex drawing where the circular order forms a 
Hamiltonian cycle \emph{closed outer $k$-planar} and \emph{closed
  outer $k$-quasi-planar}, respectively.  We observe that the
property stated by Hong and Nagamochi carries over to general
outer $k$-planar and outer $k$-quasi-planar graphs.

\begin{observation}[\!\!\cite{hn-ltatfo2p-DAM19}]
  \label{obs:biconnected-components}
  A graph $G$ is full outer $k$-planar (full outer $k$-quasi-planar)
  if and only if its biconnected components are closed outer $k$-planar
  (closed outer $k$-quasi-planar).
\end{observation}
This observation follows from the fact that each biconnected component
of a full outer $k$-planar (full outer $k$-quasi-planar) drawing is
bounded by a cycle, and thus, is a closed outer $k$-planar (closed
outer $k$-quasi-planar) graph.

We show that we can encode closed outer $k$-planarity and 
closed outer $k$-quasi-planarity using Monadic Second-Order Logic (MSO$_2$). 
To do so, we initially give a brief introduction to MSO$_2$ and
Courcelle's theorem.
Then we design MSO$_2$ formulas expressing crossing patterns of closed $k$-planar and closed $k$-quasi-planar drawings. 
Using Observation~\ref{obs:biconnected-components}, we can test
the full outer $k$-planarity of a graph by testing its biconnected
components for closed outer $k$-planarity
using the MSO$_2$ formulas.  This, together with Courcelle's theorem
(see Theorem~\ref{thm:courcelle} below) and the fact that outer $k$-planar
graphs have bounded treewidth (see~\cite[Proposition~8.5]{Wood2007}
or~\cite{fgkow-btokp-GD24})
yields a linear-time algorithm for testing full outer $k$-planarity.
Unfortunately, grid graphs, which are outer quasi-planar, have unbounded
treewidth.  So the same approach does not yield a linear-time
algorithm for testing full outer $k$-quasi-planarity.

Monadic Second-Order Logic (MSO$_2$)~-- a subset of \emph{second-order
  logic}~-- can be used to express certain graph properties.
It is built from the following primitives.
\begin{itemize}
\item variables for vertices, edges, sets of vertices, and sets of edges;
\item binary relations for: equality ($=$), membership in a set ($\in$), subset of a set ($\subseteq$),
 and edge--vertex incidence ($I$);
\item standard propositional logic operators: $\lnot$, $\land$, 
$\lor$, $\rightarrow$, and $\leftrightarrow$.
\item standard quantifiers ($\forall,\exists$) which can be applied to all types of variables.
\end{itemize}
% For a graph $G$ and an MSO$_2$ formula~$\phi$, we use $G \models \phi$
% to indicate that $\phi$ can be satisfied by $G$ in the obvious way.
Properties expressed in this logic allow us to use the powerful
algorithmic result of Courcelle stated next. 

\begin{theorem}[\!\!\cite{courcelle1990,courcelle2012}]
  \label{thm:courcelle}
  For any integer $t \geq 0$ and any MSO$_2$ formula $\phi$ of length
  $\ell$, an algorithm can be constructed that takes a graph $G$ and
  decides whether $G$ satisfies $\phi$.  If~$G$ has $n$ vertices, $m$
  edges, and treewidth at most $t$, then the algorithm runs in time
  $O(f(t,\ell)\cdot (n+m))$, where the function $f$ is computable.
\end{theorem}

The challenge in expressing outer $k$-planarity or outer $k$-quasi-planarity in MSO$_2$ is that MSO$_2$ does
not allow quantification over sets of pairs of vertices which involve non-edges.
Namely, it is unclear how to express a set of pairs that forms the
circular order of vertices on the boundary of our convex drawing.
However, if this circular order forms a Hamiltonian cycle in 
our graph, i.e., the given graph is closed, then
we can indeed express this in MSO$_2$.
With the edge set of a Hamiltonian cycle of our graph
in hand, we can then ask that this cycle was chosen in such a way that
the other edges satisfy either $k$-planarity or $k$-quasi-planarity.

The MSO$_2$ formulas presented below assume that a graph $G$ is given
and uses the edges, vertices, and incidences of~$G$.  In the
following, $V$ is the vertex set of~$G$, $u,v \in V$, and
$U \subseteq V$.  Further, $E$ is the edge set of~$G$, $e,f \in E$,
and $F \subseteq E$.  (We also use sub- and superscripted variants of
these variables.)  In addition to the quantifiers above,
we use~$\exists^{=x}$ as shorthand to express
the existence of exactly $x$ (but no $x+1$) pairwise different
elements satisfying a property.

The first formula expresses the connectivity of a subgraph induced by
a given edge set~$F$.

\bigskip

\noindent$\textsc{Connected-Edges}(F)
\equiv
(\forall F') \big([F' \neq \emptyset \wedge F' \subsetneq F] \rightarrow$\\
\noindent\phantom{$\textsc{Connected-Edges}(F)\equiv(\forall F')\big($}%
$(\exists e,f,v)[e \in F' \wedge f \in F \setminus F' \wedge v \in V
\wedge I(e,v) \wedge I(f,v)]\big)$
  
\bigskip

It states that, for every nonempty proper subset $F'$ of the given
edge set $F$, we can find an edge~$e$ in~$F'$, an edge~$f$
in~$F \setminus F'$, and a vertex that is incident to both.

We use the predicate \textsc{Connected-Edges}$(F)$ and the following
predicates to express the Hamiltonicity of~$G$.  The predicate
\textsc{Cycle-Set}$(F)$ expresses that $F$ forms a set of cycles,
\textsc{Cycle}$(F)$ expresses that $F$ consists of a single cycle, and
\textsc{Span}$(F)$ forces $F$ to span~$G$.
\begin{align*}
 \textsc{Cycle-Set}(F) &\equiv 
 (\forall e)\Big[ e\in F \rightarrow (\exists^{=2} f)\big[ f\in F 
\wedge e \neq f \wedge  (\exists v) [I(e,v) \wedge I(f,v)]\big]\Big] 
\\
 \textsc{Cycle}(F) &\equiv \textsc{Cycle-Set}(F) \wedge \textsc{Connected-Edges}(F)\\
 \textsc{Span}(F) &\equiv (\forall v)(\exists e)[e\in F \wedge  I(e,v)]\\
 \textsc{Hamiltonian}(F) &\equiv [\textsc{Cycle}(F) \wedge \textsc{Span}(F)]
\end{align*}
% In addition, we use a predicate 
% $\textsc{Connected}(A,E^*)$ to indicate that
% a subset of edges $E^*$ spans a subset of vertices $A$.
% $$
% \textsc{Connected}(A,E^*) = \textsc{Connected-Edges}(E^*) \wedge 
% (\forall v \in A)(\exists e\in E^*)[e\in E^* \wedge  I(e,v)].
% $$

The following predicate \textsc{Vertex-Partition} expresses the
existence of a partition of a set $U$ of vertices into $k$ disjoint
subsets.
\[\textsc{Vertex-Partition}(U,U_1,\dots,U_k) \equiv(\forall v \in U) 
  \! \left[\left( \bigvee_{i=1}^k v \in U_k \right)
  \!\wedge\! \left( \bigwedge_{i \neq j} \neg (v \in U_i \wedge v \in U_j)
  \right) \right]\]
Using \textsc{Vertex-Partition}, we can formulate
\begin{align*}
  \textsc{Connected}(F,U) \equiv
  & \;(\forall U_1,U_2 \subseteq U) \Big(\textsc{Vertex-Partition}(U,U_1,U_2) \\
  & \;\land \big(\exists e \in F, u_1 \in U_1, u_2 \in U_2 \big)
    \big(I(e,u_1) \land I(e,u_2) \big) \Big),
\end{align*}
which is true if and only if the given set~$U$ of vertices is
connected via the given set~$F$ of edges.

For a closed outer $k$-planar or closed outer $k$-quasi-planar
graph~$G$, we want to express that two edges~$e$ and~$e_i$ cross.  To
this end, we assume that $G$ contains a Hamiltonian cycle~$E^*$
that defines the outer face.  We partition the vertices of $G$ into
three subsets~$C$, $A$, and~$B$, as follows: $C$ is the set containing
the endpoints of~$e$, whereas $A$ and~$B$ are connected subgraphs on
the remaining vertices that use only edges of~$E^*$.  In this way, we
partition the vertices of~$G$ into two sets, one left of~$e$ and 
one right of~$e$.  For such a partition, $e_i$ must cross~$e$ whenever
$e_i$ has one endpoint in~$A$ and one in~$B$.
\begin{align*}
 \textsc{Cro}&\textsc{ssing}(E^*,e,e_i)\equiv (\forall A,B,C) \big[ \big(\textsc{Vertex-Partition}(V,A,B,C)\\
 &\wedge (I(e,x) \leftrightarrow x \in C) \wedge 
\textsc{Connected}(A,E^*) \wedge \textsc{Connected}(B,E^*)\big)\\
 &\rightarrow (\exists a \in A)(\exists b \in B )[I(e_i,a) \wedge
 I(e_i,b)] \big]
\end{align*}
Now we can describe the allowed crossing patterns.  We first express
closed outer $k$-planarity:
\[\textsc{Closed Outer $k$-Planar}_G \equiv (\exists E^*)\Big[\textsc{Hamiltonian}(E^*) \wedge \] \[(\forall e)
\Big[(\forall e_1,\dots,e_{k+1}) \Big[ \Big(\bigwedge_{i=1}^{k+1} e_i \neq e \wedge \bigwedge_{i\neq j} e_i \neq e_j \Big) \rightarrow \bigvee_{i=1}^{k+1} \neg \textsc{Crossing}(E^*,e,e_i)\Big] \Big]\Big]\]

Here we insist that $G$ is Hamiltonian and that, for every edge $e$
and any set of $k+1$ distinct other edges, at least one among them
does not cross~$e$.

Now we express closed outer $k$-quasi-planarity:
 \[\textsc{Closed Outer $k$-Quasi-Planar}_G \equiv (\exists E^*)\Big[\textsc{Hamiltonian}(E^*)\wedge \]
 \[(\forall e_1,\dots,e_k) \Big[ \Big( \bigwedge_{i\neq j} e_i \neq e_j \Big) \rightarrow \bigvee_{i\neq j} \neg \textsc{Crossing}(E^*,e_i,e_j)\Big]\Big]\]

Again, we insist that $G$ is Hamiltonian and further that, for any set
of $k$ distinct edges, there is at least one pair among them that does
not cross. 

The formulas above give us the following.
\begin{theorem}\label{thm:mso2}
Closed outer $k$-planarity and closed outer $k$-quasi-planarity can be expressed in MSO$_2$ with a formula whose size depends only on~$k$.
\end{theorem}

Using the fact that outer $k$-planar graphs have bounded treewidth
yields the following recognition result.
\begin{theorem}
  \label{lem:recfull}
  We can test whether a graph is full outer $k$-planar in time
  linear in the size of the graph.
\end{theorem}
\begin{proof}
Recall that in a full outer $k$-planar 
drawing there is no crossing on the outer boundary of the
drawing and 
each biconnected component of the graph with such a drawing is
a closed outer $k$-planar graph (Observation~\ref{obs:biconnected-components}).
Thus, in order to test full outer $k$-planarity of a given graph~$G$,
it suffices to test whether each of its biconnected components admits a 
closed outer $k$-planar drawing.
We can break up~$G$ into biconnected components by
obtaining the set of cutvertices~\cite{cut-vertices}.
This takes linear time.
Checking each biconnected component for closed
outer $k$-planarity can be done
via the above MSO$_2$ formula in time linear in the size
of the component; see Theorems~\ref{thm:courcelle} and~\ref{thm:mso2}.
The formula also guarantees that the Hamiltonian cycle 
(if present) is placed on the outer boundary of the drawing of each 
component.
We put the individual drawings of the components together by
reidentifying the cutvertices and without introducing any crossings.
This can also be done in linear time.
Thus, the total runtime is linear in the size of the input graph~$G$.
\end{proof}

\section{Discussions and Open Problems}

Every planar graph is outer $5$-quasi-planar because
planar graphs have page number~4 \cite{YANNAKAKIS198936}.  (Planar
graphs that require four pages have been discovered
recently~\cite{yannakakis-2020,BekosKKPRU20}).
There are also planar graphs that are not outer quasi-planar.  The following question is still open.

\begin{question}
Is every planar graph outer 4-quasi-planar?
\end{question}

We now discuss the relation between our crossing-restricted convex 
drawings and 
the class of intersection graphs of chords of a circle, i.e., 
\emph{circle graphs}. 
Such representations are called \emph{chord diagrams}.
Here, a convex drawing $D$ of a graph $G$ can be seen as a chord 
diagram
and as such provides a corresponding graph $H$ where each adjacency between 
two vertices corresponds to a
crossing between the edges of our drawing.
Independent sets in $H$ correspond to collections of pairwise 
non-crossing edges in $D$, i.e.,
outerplanar sub-drawings of~$D$. Thus, $k$-coloring~$H$ corresponds to 
partitioning~$D$ into 
edge sets $E_1, \ldots, E_k$ such that each sub-drawing of~$D$ formed 
by the edges of $E_i$ is outerplanar. That is, the partition $E_1, \ldots, E_k$ forms a book embedding of $G$ with $k$ pages. 
So, $k$-coloring the chord diagram provides a $k$-page book embedding 
of $G$.
Interestingly, it is NP-complete to test whether a chord diagram can 
be 4-colored~\cite{GJMP80}, but
testing whether it can be 3-colored is still open~\cite{wiki:circle_graphs}. 
On the other hand, circle graphs are 
\emph{$\chi$-bounded}~\cite{KOSTOCHKA97}, i.e., the chromatic number 
$\chi(G)$ of a circle graph $G$ is bounded by a function of 
the \emph{clique number} $\omega(G)$ of $G$, that is, the number of
vertices in the maximum clique of~$G$.
Until recently the best known bound was $7\omega^2$ due to Davis et
al.~\cite{dm-cgaqxb-BLM21},
but then Davis showed~\cite{davis22} an improved bound of
$O(\omega \log \omega)$, which is asymptotically tight.  In 
particular, this means that every outer $k$-quasi-planar drawing can 
be partitioned into 
$O(k \log k)$ pages (since we cannot have $k$ mutually 
crossing edges, i.e., there is no $k$-clique in the corresponding 
intersection graph). 
For quasi-planar graphs ($k = 3$) there is a tighter bound. 
Ageev~\cite{A96} showed that any triangle-free circle graph has 
chromatic number at most~5. 
Because for a fixed drawing of an outer quasi-planar graph its
corresponding circle graph is triangle-free, it has chromatic number
at most~5, and thus, we can embed the outer quasi-planar graph in a
book with five pages.
An immediate open question is to improve this bound on the page number.

Ageev \cite{A96} constructed a triangle-free circle graph
$G_\mathrm{ageev}$ with $\chi(G_\mathrm{ageev}) = 5$.
The drawing of the outer quasi-planar graph~$G$ 
corresponding to the circle graph $G_\mathrm{ageev}$
cannot be embedded on four pages because 
the circle graph has
chromatic number~5. It turns out, however, that there exists a linear 
order of the vertices under which $G$ can be embedded on four pages,
even if we add edges to make it
maximal, but there does not exist such an order with the additional
property that the drawing is outer quasi-planar.  We have verified
this by constructing a logical formula that tests outer
quasi-planarity (Appendix~\ref{sec:oqpchecker}) and 4-page
embeddability (Appendix~\ref{sec:pnchecker}) at the same time.

\begin{question}
  \label{quest:ageev}
  Does every outer quasi-planar graph have page number at most~4?
\end{question}

We also point to the gap between the lower bound (see Lemma~\ref{obs:largest_clique}) and the upper bound (see Theorem~\ref{thm:maxmindeg}) for the degeneracy of outer $k$-planar graphs.

\begin{question}
Can we improve the lower bound of ${\lfloor\sqrt{4k+1}\rfloor + 1}$ or the upper bound of $\lfloor 3.5\sqrt{k}\rfloor$ on the degeneracy of outer $k$-planar graphs?
\end{question}

% There are several interesting questions from an algorithmic point of view. 
%
% \begin{question}
%   \label{ques:quasipol-pol}
% Can we improve the quasi-polynomial algorithm in Theorem~\ref{thm:recognition-algorithm} to a polynomial one?
% \end{question}

The linear runtime algorithm for testing full outer $k$-planarity in
Theorem~\ref{lem:recfull} relies on Courcelle's machinery for solving
MSO$_2$ formulas and, therefore, has a notoriously bad runtime
dependence (of $2^{O(k^2)}$) on the parameter~$k$.  Hence, in light of
a similar result for one-page crossing
minimization~\cite{kot-aifpafopcm-IPEC17}, it is natural to ask
whether this can be improved.
\begin{question}
  \label{ques:dp}
  Is there an explicit dynamic programming algorithm to decide whether
  a graph $G$ is full outer $k$-planar with a better runtime
  dependence on the parameter~$k$?
\end{question}

% \begin{question}
%   Going beyond Questions~\ref{ques:quasipol-pol} and~\ref{ques:dp},
%   can we obtain an XP or a fixed-parameter tractable algorithm for
%   testing general outer $k$-planarity?
% \end{question}

Last but not least:
\begin{question}
What is the complexity of outer $k$-quasi-planarity testing? 
\end{question}
In general $k$-quasi-planarity testing has notoriously eluded the
efforts to establish its computational complexity even for $k=3$ (or
simply, quasi-planarity testing, with respect to our
definition).  Recently, Angelini et
al.~\cite{albfp-2lqpohccspqrt-SODA21} showed that testing 2-level
quasi-planarity (i.e., quasi-planarity for a bipartite graph in a
2-layer layout, where the vertices of each partition are on one of two
parallel lines and the edges are drawn between the lines) is
NP-complete, making first progress in this direction.  Can this be
adapted to show NP-completeness for testing outer $k$-quasi-planarity
for any $k\ge3$?

\bibliographystyle{plainurl}
\bibliography{abbrv,refs}

\appendix
\clearpage

\section{SAT Formulations}
\label{sec:satchecker}

In the following two sections we describe SAT formulations that can be
used to test whether a given graph is outer quasi-planar
(Section~\ref{sec:oqpchecker}) and to compute its page
number (Section~\ref{sec:pnchecker}).  We present the formulas in
first-order logic.  After transformation to Boolean logic, the
resulting formulas can be solved using, e.g., MiniSat
(see \url{http://minisat.se/}).

\subsection{Outer Quasi-Planarity Checker}
\label{sec:oqpchecker}

In this section, we describe a logical formula for testing whether a
given graph~$G$ is outer quasi-planar (that is, outer
3-quasi-planar).
We used this formula in order to prove
Propositions~\ref{prop:someoqp}(d) and~\ref{prop:somenotoqp} in
Section~\ref{sec:o-kqp}.
An outer quasi-planar embedding corresponds to a
circular order of the vertices.
If we cut a circular order at some vertex to turn the circular into a
linear order, the edge crossing pattern remains the same.
Therefore, we look for a linear order.
For any pair $(u,v)$ of vertices of~$G$ with $u \ne v$, we introduce a
Boolean variable~$x_{u, v}$ that expresses whether vertex~$u$
is before~$v$ in the linear order.  In addition, for any pair $(e,e')$
of edges of~$G$ with $e \ne e'$, we introduce a Boolean
variable~$y_{e, e'}$ that expresses whether edge~$e$ crosses
edge~$e'$.  Now we list the clauses of our SAT formula.
\begin{align}
  \label{s_con:transitivity}
  & x_{u,v} \wedge x_{v,w} \Rightarrow x_{u,w}
  && \text{for each } u \ne v \ne w \ne u \in V(G);\\
  \label{s_con:equivalence}
  & x_{u,v} \Leftrightarrow \neg x_{v,u}
  && \text{for each } u\neq v \in V(G);\\
  \label{s_con:edgescross}
  & x_{u,u'} \wedge x_{u',v} \wedge x_{v,v'} \Rightarrow y_{e, e'}
  && \text{for each } e=(u,v) \;\ne\; e'=(u',v') \in E(G); \\
  \label{s_con:edgesdontcross}
  & \neg (y_{e_1,e_2} \wedge y_{e_1,e_3} \wedge y_{e_2,e_3})
  && \text{for each } e_1, e_2, e_3 \in E(G) \text{ with different endpoints.}
\end{align}
The first two sets of clauses describe the linear order:
transitivity~\eqref{s_con:transitivity} and
anti-symmetry~\eqref{s_con:equivalence}.
Clauses~\eqref{s_con:edgescross} realize the
intended meaning of the variables~$y_{e, e'}$.  Finally,
clauses~\eqref{s_con:edgesdontcross} ensure that no
three edges pairwise cross.

\subsection{Page Number Checker}
\label{sec:pnchecker}

In this section we provide a SAT formula that, given a graph~$G$
and an integer $k>0$, has a satisfying truth assignment if and only if
$G$ has page number at most~$k$.
A similar SAT formulation has been implemented by
Pupyrev~\cite{p-mllpg-GD17} (see \url{http://be.cs.arizona.edu/}).
For completeness, we list the constraints that we used in order to
compute the page number of~$G_\mathrm{ageev}$ (see
Question~\ref{quest:ageev}).  We find a linear order of the vertices
that corresponds to a $k$-page embedding.
For every pair $(u,v)$ of vertices of $G$ with $u \ne v$, we introduce
a Boolean variable~$x_{u, v}$ (as in Section~\ref{sec:oqpchecker})
that expresses whether $u$ is before $v$ in the linear order.
For every edge~$e$ of~$G$ and every page $i\in\mathcal{P}=\{1, \dots, k\}$,
we introduce a Boolean variable~$p_{i, e}$ that expresses whether
edge~$e$ is on page~$i$.  Now we list the clauses of our SAT formula.
\begin{align}
  \label{s_con:pn_transitivity}
  & (x_{u,v} \wedge x_{v,w}) \Rightarrow x_{u,w}
  && \text{ for each } u \ne v \ne w \ne u \in V(G);\\
  \label{s_con:pn_equivalence}
  & x_{u,v} \Leftrightarrow \neg x_{v,u}
  && \text{ for each } u\neq v \in V(G);\\
  \label{s_con:pn_edgeonsomepage}
  & \bigvee_{i \in \mathcal{P}} p_{i,e}
  && \text{ for each } e \in E(G);\\
  % \label{s_con:pn_edgeuniquepage}
  % & \neg(p_{i,e} \wedge p_{j,e})
  % && \text{ for each } i\neq j \in \mathcal{P}, e\in E(G);\\
  \label{s_con:pn_edgesdontcrossonpage}
  & (p_{i,e} \wedge p_{i,e'}) \Rightarrow
  \neg(x_{u,u'} \wedge x_{u',v} \wedge x_{v,v'})
  && \text{ for each } i \in \mathcal{P}, u \ne u', v \ne v' \in V(G), \\
  & && \quad e=(u,v), e'=(u',v') \in E(G)
\end{align}
The first two sets of clauses,
\eqref{s_con:pn_transitivity}--\eqref{s_con:pn_equivalence}, are the
same as clauses~\eqref{s_con:transitivity}--\eqref{s_con:equivalence}
since they describe the linear order.
Clauses~\eqref{s_con:pn_edgeonsomepage} %--\eqref{s_con:pn_edgeuniquepage}
guarantee that every edge is on some page.
% (Clause~\eqref{s_con:pn_edgeuniquepage} is not strictly necessary.)
Clauses~\eqref{s_con:pn_edgesdontcrossonpage} ensure that two edges
that have different endpoints and lie on the same page do not cross.
(If two edges share an endpoint, they cannot cross.)

\end{document}